\documentclass[graybox]{svmult}

\usepackage{mathptmx}      
\usepackage{amsmath,amssymb}
\usepackage{graphicx}
\usepackage{cite}          
\usepackage{url}

\numberwithin{equation}{section}

\newcommand{\boldgreek}[1]{\mbox{\boldmath$#1$}}
\newcommand{\R}{I\kern-0.37emR}
\newcommand{\ny}{n\rightarrow\infty}

\begin{document}


\title*{Estimation of the Risk Measure\\ under a Nuisance Autoregression}
\author{Jana Jure\v{c}kov\' a \and Jan Picek}

\institute{Jana Jure\v{c}kov\' a \at
  Institute of Information Theory and Automation, Czech Academy of Sciences,\\
	Prague, Czech Republic,
  \email{jureckova@utia.cas.cz}
  \and
  Jan Picek \at
  Technical University of Liberec, Faculty of Science, Humanities and Education,\\ Liberec, Czech Republic,
  \email{jan.picek@tul.cz}
}

\maketitle
\abstract{
The goal of an experiment is to evaluate the profit, loss, or the amount of a physical entity over a period. The measurements 
$X_t$ can be influenced by the  values measured in the past; hence we describe the situation with an autoregression model, whose
autoregression coefficients are generally unknown. The variable of interest is the error term $Z_t$ of the model, which is the increment of 
$X_t$ with respect to the past, but itself unobservable.  The problem is to estimate various quantile functions of $Z$, as the risk measure of 
the loss or the related economic indicators.  
We construct an estimate of quantile functions of $Z$ in the situation that the inference 
is possible only by means of observations $X$. The proposed estimates are based on the R-estimators of autoregression coefficients, 
combined with the autoregression quantiles.  
}

\section{Introduction}

In the everyday practice we monitor various  indicators,  as the environmental, industrial, health, physical, economic  and others, 
in order to make a decision or to plan the next steps. These indicators measure the financial or health risk, the environmental and technical risk, 
and can also describe the income inequalities between the social groups. One meets them in the water management, in insurance, and elsewhere.
We monitor the values $X_t$, 
  related to their past values through the $p$-order autoregressive model (AR(p))
\begin{equation}
 X_{t}=\phi _{1}X_{t-1}+...+\phi _{p}X_{t-p}+Z_{t}, \quad t\in \mathbb{Z}
\label{1} 
\end{equation}
where the observations ${\mathbf X}_n=(X_{t-1},\ldots,X_{t-p})^{\top}, \; t=1,\ldots,n$ denote a gain, profit of some sort, or the loss and shortage, 
and $Z_t$ is an increment in gain at time $t$, with respect to period $[t-1,\ldots, t-p]$, under the validity of model (\ref{1}). 
The autoregression parameters  $\boldgreek\phi=(\phi_1,\ldots,\phi_p)$ are unknown, hence the values $Z_t, \; t=0,\pm1, \pm 2, \ldots$  are not
observable.

We wish to follow the behavior of the indicators, which  can be mathematically described as the quantile functionals of $Z$.\\
We usually assume that $Z$ has an absolutely continuous distribution function $F$ and quantile function 
$Q(\alpha)=F^{-1}(\alpha)  ; 0<\alpha<1,$
but their explicit forms are generally unknown. The problem is to estimate the functionals of $Q$, which are often the  
risk measures and other economic, industrial, physical, but also health and environmental indicators. 
We estimate such functionals, based on independently repeated observations $X_t, \; t=1,\ldots,n.$
The experiment is organized  so that the components  $Z_{t}, t=2,\ldots,n$, the increments of the loss, shortage etc., are independent. 
We assume that they are 
identically distributed (iid) according to a continuous distribution function $F$ with the quantile function $Q=F^{-1}$. They are generally unknown, 
but $F$ is assumed to increase on the set $\{ z : 0<F(z)<1\}$, where it satisfies the following conditions:
\begin{description}
\item[(A1)] The distribution function $F$ has a continuous density $f$, that is positive on the support of $F.$
\item[(A2)] $E(Z_{t}) =0,$  $Var(Z_{t})=\sigma ^{2}<\infty $,  $0< \int z^{4}dF(z)<\infty. $
\item[(A3)] We assume  that $\phi_{j}$ are such that all roots of the equation\\
$ z^{p}-\phi _{1}z^{p-1}-\phi _{2}z^{p-2}-...-\phi _{p}=0 $
are inside the unit circle.
\end{description}
Denote
\begin{eqnarray}\label{Y}
&&\mathbf{Y}_{t-1}=\left(X_{t-1},...,X_{t-p}\right)^{\top} \; \mbox{ and}\\ 
&&\mathbf{Y}_{t-1}^{\ast }=\left(1,X_{t-1},...,X_{t-p}\right) ^{\top}.\nonumber
\end{eqnarray}
 Moreover, let  $\boldgreek{\Sigma}_{n}$ and
$\boldgreek{\Sigma}_{n}^{\ast}$ denote the following matrices of the respective orders $n\times p$  and $n\times(p+1)$ :
\begin{equation}\label{Sigma}
\boldgreek{\Sigma}_n=\frac 1n \sum_{t=1}^n \mathbf{Y}_{t-1}\mathbf{Y}_{t-1}^{\top }, \qquad
\boldgreek{\Sigma}_n^*=\frac 1n \sum_{t=1}^n \mathbf{Y}_{t-1}^{\ast }\mathbf{Y}_{t-1}^{\ast\top }.
\end{equation}
The problem is to estimate a  functional $\mathcal S(Q)$ of the quantile function $Q$.
The interest can be also in the confidence intervals for $\mathcal S(Q),$ or in a comparison of two functionals, 
representing two possible treatments. Specifically,  $\mathcal S(Q)$ can be a risk measure, but  also an energy consumption 
over a period, water flow over a period, etc. 
Notice that if  the autoregression or other sequential dependence in the data is ignored, it can mask a possible trend in the 
measurements. Conversely, it can indicate a trend,
even when there is no one. This is worried e.g. by hydrologists (see \cite{YuePilon1} and \cite{YuePilon2} ). Hence, 
our  goal is estimating a linear functional of the quantile function of $Z$ 
even in the situation that $Z$  is unobservable, while its values are affected by the autocovariances. 
Our main mathematical tool in solving this situation is the \textit{autoregression quantile } for model (\ref{1}), combined with an \textbf{R-estimate}
of the autocovariance parameters. 
These tools are studied in \cite{Arslan}, \cite{Koul1992}, \cite{mukherjee}, \cite{Sankhya2005}, besides others. 

\section{Autoregression quantile and R-estimation}
\setcounter{equation}{3} 
Let us  work with  autoregressive model (\ref{1}) of the finite  assumed order $p$.
The $\alpha$-autoregression quantile for autoregressive model was first considered in \cite{KoulSaleh1995}, 
and later it was studied in \cite{HallinJurec1999}, \cite{Hallinetal2007} and \cite{Arslan}, among others. 
Keeping the notation of (\ref{Y}) and (\ref{Sigma}), we remind that
the $\alpha $-th autoregression quantile 
$\widehat{\boldgreek{\phi }}(\alpha )=\left(\widehat{\phi}_0(\alpha), \widehat{\phi}_1(\alpha),\ldots,\widehat{\phi}_p(\alpha)\right)^{\top} 
=\left(\widehat{\phi }_{0}(\alpha),\widehat{\boldgreek{\phi }}^{\ast\top}(\alpha)\right)^{\top}$ %
is defined as a solution of the minimization 
\begin{equation}\label{5}
\widehat{\boldgreek{\phi}}(\alpha )=\arg \min_{\mathbf{b\in\mathbb{R}}^{p+1}}\sum_{t=1}^{n}\rho_{\alpha }\left( X_{t}
-\mathbf{Y}_{t-1}^{*\top}\mathbf{b}\right)  
\end{equation}%
where $$\rho_{\alpha }(u) =\vert u\vert \left \{ \alpha \mathcal I( u>0) +(1-\alpha) \mathcal I( u<0) \right \},  \;  
u\in\R$$ for  $\alpha \in ( 0,1) $, where $\mathcal I( \cdot) $ is the indicator function. 

Koul and Saleh  \cite{KoulSaleh1995} derived the following asymptotic representation of  the $\alpha $-autoregression quantile under $\ny$:
\begin{eqnarray}
\label{36}
&&\widehat{\boldgreek\phi}_n(\alpha)%
-(F^{-1}(\alpha),\phi_1,\ldots,\phi_p)^{\top}=\\
&&n^{-1}({\boldgreek\Sigma}^*_n)^{-1}(f(F^{-1}(\alpha)))^{-1}%
\sum_{t=1}^n{\bf Y}^*_{t-1}\left(\alpha-\mathcal I[Z_t\leq F^{-1}(\alpha)]\right)+o_p(n^{-\frac 12})\nonumber
\end{eqnarray}
as $\ny,$ and the convergence is uniform over each subinterval
$[\alpha^*,1-\alpha^*]\subset (0,1).$

Among the quantile functionals of $Q$, typical are the measures of the risk in various contexts, as  in finances,  in health problems, environmental 
and technical risks. Moreover, the economists often analyze various welfare measurement functionals. Various welfare measurement functionals, 
often used in the economic and sociological studies, are described in \cite{Gastwirth2016}, \cite{Gastwirth2025},
along with the income inequalities. We shall illustrate the family of quantile functionals on the  \textsl{Conditional Value-at-Risk} $\sf CVaR$, 
a popular risk measure, also  known as the \textit{expected shortfall}. It is equal to 
\begin{equation}\label{1a}
\sf{CVaR}_{\alpha}(Z)= E\{Z|Z>Q(\alpha)\}=(1-\alpha)^{-1}\int_{\alpha}^1 Q(u)du=(1-\alpha)^{-1}\int_{Q(\alpha)}^{\infty}zdF(z).
\end{equation}
Before a detailed analysis of the quantile functionals, let us first describe  the concepts 
of the \textit{R-estimation of AR parameters}, and of its combination with the autoregression quantile. 

\subsection{R-estimation of autoregression parameters} 
The following version of the R-estimate of autoregression parameters $\boldgreek\phi$ was proposed in \cite{Koul1992} and studied in \cite{mukherjee}. 
There are alternative versions 
of R-estimator,  mutually asymptotically equivalent, si\-mi\-lar\-ly as in the linear model.  
Consider the model (\ref{1}) 
with the observations  $\mathbf Y_{t-1}=(X_{t-1}\ldots, X_{t-p})^{\top}, \; t=1,\ldots,n$, unknown autoregression parameters 
$\boldgreek{\phi}= (\phi_1,\ldots,\phi_p)^{\top}$ and \textit{i.i.d.} unobservable innovations $Z_1,\ldots,Z_n.$ The distribution function $F$ 
of $Z_t$ is unknown, but assumed to satisfy the conditions $\textbf (A1) -\textbf(A3).$
For $\mathbf b\in \R^p$, let $R_{1,\mathbf b}\ldots,R_{n,\mathbf b}$ denote the ranks of the residuals 
$X_{t}-\mathbf b^{\top}\mathbf Y_{t-1}, \; t=1,\ldots, n.$ We choose a nondecreasing  score function $J: (0,1)\mapsto\R^1$, 
\begin{equation}
\label{8}
J_{\lambda}(u)=\lambda-\mathcal I[u<\lambda],  \; \ 0<u<1
\end{equation}
with any fixed $\lambda\in(0,1).$ Parallelly, we can use another nondecreasing square-integrable score function on $(0,1)$. 
The estimator of $\boldgreek{\phi}$, denoted   $\widetilde{\boldgreek\phi}_{nR}(\lambda)$, 
is defined  as a  minimizer, over ${\mathbf b}\in{\R}^{p}$, of  the Jaeckel \cite{Jaeckel1972} measure of dispersion 
of rank residuals criterion ${\cal D}_n$,   
\begin{equation}
\label{11a}
{\mathcal D}_n({\mathbf b})=\sum_{t=1}^n(X_{t}-\mathbf b^{\top}\mathbf Y_{t-1})%
\left[J_{\lambda}\left(\frac{R_{nt}(X_{t}-\mathbf b^{\top}\mathbf Y_{t-1})}{n+1}\right)-\overline{J}_{n\lambda}\right].
\end{equation}
An alternative version of R-estimation of $\boldgreek\phi$ is proposed in \cite{Koul1992}; both versions are asymptotically equivalent.
Because the ranks are invariant to the shift in location, the same is true for the R-estimator. 
It is shown in \cite{Ossiander} that, under {\bf (A1)}-{\bf(A3)},
\begin{equation}
\label{40}
n^{-1}\sum_{t=1}^n({\bf Y}_{t-1}-\overline{\mathbf Y})%
({\bf Y}_{t-1}-\overline{\mathbf Y})^{\top}\stackrel{p}{\longrightarrow}%
\tilde{\boldgreek\Sigma} \quad \mbox{ as } \ \ny,
\end{equation}
where 
$\tilde{\boldgreek\Sigma}$ is a positively definite $p\times p$ matrix, and
\begin{equation}
\label{41}
\widetilde{\boldgreek\phi}_{n}-{\boldgreek\phi}
=n^{-1}\tilde{\boldgreek\Sigma}^{-1}(f(F^{-1}(\alpha)))^{-1}%
\sum_{t=1}^n({\mathbf Y}_{t-1}-\overline{\mathbf Y})\left(\alpha-\mathcal I[Z_t\leq F^{-1}(\alpha)]\right)+o_p(n^{-\frac 12}).
\end{equation}
Let now $\tilde{\phi}_{n0}(\alpha)$ be the $[n\alpha]$-th order statistic of the residuals
$X_t-{\mathbf Y}_{t-1}^{\top}\widetilde{\boldgreek\phi}_n, \;  t=1,\ldots,n$, for $\alpha$ running over $(0,1)$.
It is the solution of the minimization
\begin{equation}
\label{14c}
\tilde{\phi}_{n0}(\alpha)=\arg\min\left\{\sum_{t=1}^n \rho_{\alpha}\left(X_t-b-{\mathbf Y}_{t-1}^{\top}\widetilde{\boldgreek\phi}_n\right), 
\quad b\in{\R}^1\right\}
\end{equation}
where $\widetilde{\boldgreek\phi}_{n}$ is the R-estimate of ${\boldgreek\phi}.$ 
Then Lemma 1.1 
in \cite{Ossiander} helps to get rid of the unknown nuisance parameters, because it shows 
\begin{equation}
\label{22c}
\sup_{|b|\leq C}\Big\{n^{-\frac 12}\Big|\sum_{t=1}^n\Big(\mathcal I[X_t-{\bf Y}_{t-1}^{\top}\widetilde{\boldgreek\phi}_{n}<Q(\alpha)%
+n^{-\frac 12}b]-\mathcal I[Z_t<Q(\alpha)+n^{-\frac 12}b]\Big)\Big|\Big\}=o_p(1).
\end{equation}
The following asymptotic representation of $\tilde{\phi}_{n0}(\alpha)$,  as  $\ny,$ valid  under the conditions 
{\bf (A1)}-{\bf (A3)} for any $\alpha\in(0,1)$ uniformly for  $\alpha$ over any subinterval of $(0,1)$, has been demonstrated in \cite{Sankhya2005}, 
with a reference to \cite{KoulSaleh1995}: 
\begin{eqnarray}\label{16c}
&&\tilde{\phi}_{n0}(\alpha)\stackrel{p}{\longrightarrow} Q(\alpha) \\
&&\tilde{\phi}_{n0}(\alpha)=Q(\alpha) 
+\left(nf(Q(\alpha))\right)^{-1}\sum_{t=1}^n\left(\alpha-\mathcal I[Z_t<Q(\alpha)]\right)+o_p(n^{-\frac 12}).\nonumber
\end{eqnarray}
\begin{remark}
The $\tilde{\phi}_{n0}(\alpha)$ itself approximates the  \textbf{VaR} measure (\textit{value at risk})  at $\alpha$, in probability,  as $\ny$. Generally, it
approximates the quantile function of $Z$ at the point $\alpha$.
\end{remark}

\section{Estimation of the Conditional Value-at-Risk}
\setcounter{equation}{13}
The typical quantile functionals are measures of risk in various contexts, as  in finances,  in health problems, environmental and technical risks. 
We shall illustrate the estimation of the risk measures and also of other quantile functionals, on the popular 
\textsl{Conditional Value-at-Risk} 
(or \textit{Expected Shortfall}) $\sf CVaR$, defined as
\begin{equation}\label{1c}
\sf{CVaR}_{\alpha}(Z)=E\{Z|Z>Q(\alpha)\}=(1-\alpha)^{-1}\int_{\alpha}^1 Q(u)du=(1-\alpha)^{-1}\int_{Q(\alpha)}^{\infty}zdF(z).
\end{equation}
The expected shortfall  is applied in many areas; for instance in the management of water supplies, in the risk management
of the social security funds, the cash flow risk measurement for non-life insurance industry, the
financial risk in the industrial areas, operational risk in the banks, and others. 

Combining (\ref{14c}), (\ref{16c}) and (\ref{22c}) with the result of Bassett et al. in \cite{Bassett2004}, we obtain the following estimator of the
expected shortfall:
\begin{theorem}
Under the conditions \textbf{(A1)}- \textbf{(A3)},
\begin{eqnarray}\label{cvar1}
&&\widehat{\sf{CVAR}}_n= 
{\left\lfloor n(1-\alpha)\right\rfloor}^{-1}\min\left\{\sum_{t=1}^n \rho_{\alpha}(X_t-{\bf Y}_{t-1}^{\top}\widetilde{\boldgreek\phi}_n-\xi): 
\xi\in\R^{1}\right\}\nonumber\\
&&+n^{-1}\sum_{t=1}^n (X_t-{\bf Y}_{t-1}^{\top}\widetilde{\boldgreek\phi}_n)\\
 &&= \sf{{CVAR}}_{\alpha}+ o_p(n^{-1/2}) \;  \mbox{ as } \; \ny.\nonumber
\end{eqnarray}
\end{theorem}
\begin{proof}
According to \cite{Bassett2004},
\begin{equation}\label{cvar}
{\sf{CVAR}}_\alpha=(1-\alpha)^{-1}\min\left\{\rho_\alpha(Z-\xi):\xi\in\R^1\right\}+E Z.
\end{equation}
Hence the estimator of the shortfall is obtained from its empirical version,  after combining (\ref{14c}), (\ref {16c}) and (\ref{22c}).
\end{proof}
Similarly we can proceed with estimation of other risk measures, or with other linear quantile functionals in the autoregression model. 

\section{Simulation Study}
\label{sec:sim}
\setcounter{equation}{16}

We investigate the finite-sample performance of the proposed estimation strategy for
tail-risk functionals of the innovation distribution in autoregressive models with
unobservable innovations. The simulation design follows the two-step construction developed
in the methodological part of the paper. The nuisance autoregression is handled by a
rank-based $R$-estimator of the AR slope parameters obtained via minimization of a
Jaeckel-type dispersion with the step score
\(
J_{\lambda}(u)=\lambda-\mathbf{1}(u<\lambda).
\)
In the present simulation design without an additional location component, the innovation
tail risk is estimated by applying the empirical $CVaR_{\alpha}$ functional directly to
the raw residuals obtained after slope estimation.

For calibration purposes, we additionally report an infeasible \emph{oracle} reference,
available only in simulations, which uses the true AR slopes in the residual construction.
This benchmark isolates the finite-sample impact of estimating the nuisance autoregression.

We consider the autoregressive model
\begin{equation}
  X_t = \phi_1 X_{t-1} + \cdots + \phi_p X_{t-p} + Z_t,
  \label{eq:ar_model_sim}
\end{equation}
with i.i.d.\ innovations $(Z_t)$. The settings include AR(1) with $\phi=0.5$ and $\phi=0.8$,
and AR(2) with $(\phi_1,\phi_2)=(0.5,-0.2)$. Sample sizes are $n\in\{100,200,500\}$.
Each series is generated with a burn-in period and only the last $n$ observations are used
for estimation. We focus on high-risk levels $\alpha\in\{0.95,0.99\}$.

To reflect ideal, heavy-tailed, and contaminated environments, we consider four innovation
distributions: Gaussian, standardized Student-$t$ with $\nu=3$, a scale-mixture
$0.9N(0,1)+0.1N(0,3^2)$, and a contamination scheme with rare large shocks.

For each simulated sample we compute the proposed plug-in estimator of $CVaR_{\alpha}(Z)$
based on raw residuals obtained from the $R$-estimated slopes. The oracle version keeps the
same construction but substitutes the true AR slopes. The experiment is repeated $R$ times
and we report empirical bias and RMSE. Target values of $CVaR_{\alpha}(Z)$ are computed
analytically in the Gaussian case and by a large Monte Carlo approximation otherwise.

The results (Tables~\ref{tab:sim_ar1_phi05}--\ref{tab:sim_ar2}) show small biases for
$\alpha=0.95$ across all models and innovation scenarios. As expected, estimation becomes
more challenging for $\alpha=0.99$, particularly under heavy-tailed, mixture, and
contaminated innovations, which leads to larger RMSEs in smaller samples. Importantly, the
feasible $R$-based estimator closely tracks the oracle reference in all settings, indicating
that the finite-sample cost of estimating the nuisance autoregression is negligible for the
proposed construction. The dominant component of estimation error thus stems from the intrinsic
finite-sample variability of tail-risk estimation itself.


\begin{table}[t]
\centering
\caption{Simulation results for AR(1) with $\phi=0.5$: bias and RMSE of
$\widehat{CVaR}_{\alpha}$ computed \emph{from raw residuals}. The feasible $R$-based
procedure is compared to the oracle reference.}
\label{tab:sim_ar1_phi05}
\begin{tabular}{lllrrrr}
\hline
n & scenario & $\alpha$ & bias (R) & RMSE (R) & bias (Oracle) & RMSE (Oracle)
 \\
\hline
100 & Contamination & 0.95 & -0.0905 & 1.8427 & -0.0886 & 1.8404 \\
200 & Contamination & 0.95 & 0.0532 & 1.2925 & 0.0544 & 1.2932 \\
500 & Contamination & 0.95 & -0.0435 & 0.8360 & -0.0426 & 0.8361 \\
100 & Mixture & 0.95 & -0.0908 & 0.7230 & -0.0896 & 0.7187 \\
200 & Mixture & 0.95 & -0.1014 & 0.5820 & -0.1008 & 0.5831 \\
500 & Mixture & 0.95 & 0.0029 & 0.3587 & 0.0028 & 0.3591 \\
100 & Normal & 0.95 & -0.0361 & 0.2678 & -0.0404 & 0.2616 \\
200 & Normal & 0.95 & -0.0180 & 0.1859 & -0.0190 & 0.1843 \\
500 & Normal & 0.95 & -0.0070 & 0.1207 & -0.0082 & 0.1195 \\
100 & t3 & 0.95 & 0.0109 & 0.8582 & 0.0159 & 0.8553 \\
200 & t3 & 0.95 & -0.0021 & 0.5705 & -0.0017 & 0.5669 \\
500 & t3 & 0.95 & -0.0029 & 0.3232 & -0.0025 & 0.3225 \\
100 & Contamination & 0.99 & -0.7646 & 2.0033 & -0.7549 & 1.9986 \\
200 & Contamination & 0.99 & -0.2978 & 1.1435 & -0.2967 & 1.1395 \\
500 & Contamination & 0.99 & -0.0746 & 0.4935 & -0.0744 & 0.4920 \\
100 & Mixture & 0.99 & -0.7307 & 1.7825 & -0.7163 & 1.7772 \\
200 & Mixture & 0.99 & -0.4333 & 1.2381 & -0.4254 & 1.2349 \\
500 & Mixture & 0.99 & -0.2784 & 0.8135 & -0.2768 & 0.8121 \\
100 & Normal & 0.99 & -0.2047 & 0.4591 & -0.2055 & 0.4638 \\
200 & Normal & 0.99 & -0.1102 & 0.3241 & -0.1207 & 0.3223 \\
500 & Normal & 0.99 & -0.0553 & 0.2115 & -0.0569 & 0.2124 \\
100 & t3 & 0.99 & -0.7127 & 2.0123 & -0.7076 & 2.0155 \\
200 & t3 & 0.99 & -0.4838 & 1.5916 & -0.4850 & 1.5904 \\
500 & t3 & 0.99 & -0.1806 & 1.1102 & -0.1799 & 1.1091 \\
\hline
\end{tabular}
\end{table}

\begin{table}[t]
\centering
\caption{Simulation results for AR(1) with $\phi=0.8$: bias and RMSE of
$\widehat{CVaR}_{\alpha}$ computed \emph{from raw residuals}. The feasible $R$-based
procedure is compared to the oracle reference.}
\label{tab:sim_ar1_phi08}
\begin{tabular}{lllrrrr}
\hline
n & scenario & $\alpha$ & bias (R) & RMSE (R) & bias (Oracle) & RMSE (Oracle)
 \\
\hline
100 & Contamination & 0.95 & -0.0906 & 1.8495 & -0.0886 & 1.8404 \\
200 & Contamination & 0.95 & 0.0527 & 1.2940 & 0.0544 & 1.2932 \\
500 & Contamination & 0.95 & -0.0436 & 0.8361 & -0.0426 & 0.8361 \\
100 & Mixture & 0.95 & -0.0956 & 0.7346 & -0.0896 & 0.7187 \\
200 & Mixture & 0.95 & -0.1005 & 0.5853 & -0.1008 & 0.5831 \\
500 & Mixture & 0.95 & 0.0027 & 0.3599 & 0.0028 & 0.3591 \\
100 & Normal & 0.95 & -0.0484 & 0.2633 & -0.0435 & 0.2450 \\
200 & Normal & 0.95 & -0.0199 & 0.1903 & -0.0190 & 0.1843 \\
500 & Normal & 0.95 & -0.0069 & 0.1207 & -0.0082 & 0.1195 \\
100 & t3 & 0.95 & 0.0112 & 0.8628 & 0.0159 & 0.8553 \\
200 & t3 & 0.95 & -0.0014 & 0.5701 & -0.0017 & 0.5669 \\
500 & t3 & 0.95 & -0.0030 & 0.3232 & -0.0025 & 0.3225 \\
100 & Contamination & 0.99 & -0.7658 & 2.0102 & -0.7549 & 1.9986 \\
200 & Contamination & 0.99 & -0.2949 & 1.1414 & -0.2967 & 1.1395 \\
500 & Contamination & 0.99 & -0.0750 & 0.4931 & -0.0744 & 0.4920 \\
100 & Mixture & 0.99 & -0.7426 & 1.7897 & -0.7163 & 1.7772 \\
200 & Mixture & 0.99 & -0.4382 & 1.2369 & -0.4254 & 1.2349 \\
500 & Mixture & 0.99 & -0.2800 & 0.8126 & -0.2768 & 0.8121 \\
100 & Normal & 0.99 & -0.1253 & 0.4542 & -0.1411 & 0.4581 \\
200 & Normal & 0.99 & -0.1153 & 0.3250 & -0.1207 & 0.3223 \\
500 & Normal & 0.99 & -0.0573 & 0.2124 & -0.0569 & 0.2124 \\
100 & t3 & 0.99 & -0.7219 & 2.0174 & -0.7076 & 2.0155 \\
200 & t3 & 0.99 & -0.4866 & 1.5915 & -0.4850 & 1.5904 \\
500 & t3 & 0.99 & -0.1815 & 1.1109 & -0.1799 & 1.1091 \\
\hline
\end{tabular}
\end{table}

\begin{table}[t]
\centering
\caption{Simulation results for AR(2) with $(\phi_1,\phi_2)=(0.5,-0.2)$: bias and RMSE of
$\widehat{CVaR}_{\alpha}$ computed \emph{from raw residuals}. The feasible $R$-based
procedure is compared to the oracle reference.}
\label{tab:sim_ar2}
\begin{tabular}{lllrrrr}
\hline
n & scenario & $\alpha$ & bias (R) & RMSE (R) & bias (Oracle) & RMSE (Oracle)
 \\
\hline
100 & Contamination & 0.95 & -0.1223 & 1.8410 & -0.1170 & 1.8402 \\
200 & Contamination & 0.95 & 0.0398 & 1.2789 & 0.0453 & 1.2857 \\
500 & Contamination & 0.95 & -0.0472 & 0.8359 & -0.0455 & 0.8365 \\
100 & Mixture & 0.95 & -0.1054 & 0.7187 & -0.1001 & 0.7178 \\
200 & Mixture & 0.95 & -0.1082 & 0.5790 & -0.1071 & 0.5844 \\
500 & Mixture & 0.95 & -0.0009 & 0.3582 & -0.0015 & 0.3591 \\
100 & Normal & 0.95 & -0.0335 & 0.2582 & -0.0425 & 0.2529 \\
200 & Normal & 0.95 & -0.0175 & 0.1861 & -0.0225 & 0.1846 \\
500 & Normal & 0.95 & -0.0075 & 0.1211 & -0.0089 & 0.1197 \\
100 & t3 & 0.95 & -0.0003 & 0.8618 & 0.0069 & 0.8572 \\
200 & t3 & 0.95 & -0.0086 & 0.5693 & -0.0063 & 0.5674 \\
500 & t3 & 0.95 & -0.0068 & 0.3223 & -0.0051 & 0.3227 \\
100 & Contamination & 0.99 & -0.7800 & 2.0015 & -0.7669 & 1.9995 \\
200 & Contamination & 0.99 & -0.3043 & 1.1439 & -0.3022 & 1.1441 \\
500 & Contamination & 0.99 & -0.0795 & 0.4986 & -0.0792 & 0.4973 \\
100 & Mixture & 0.99 & -0.7469 & 1.7758 & -0.7246 & 1.7739 \\
200 & Mixture & 0.99 & -0.4415 & 1.2380 & -0.4313 & 1.2368 \\
500 & Mixture & 0.99 & -0.2822 & 0.8146 & -0.2787 & 0.8130 \\
100 & Normal & 0.99 & -0.1657 & 0.5024 & -0.1858 & 0.4990 \\
200 & Normal & 0.99 & -0.1096 & 0.3261 & -0.1221 & 0.3234 \\
500 & Normal & 0.99 & -0.0551 & 0.2094 & -0.0582 & 0.2125 \\
100 & t3 & 0.99 & -0.7218 & 2.0238 & -0.7112 & 2.0195 \\
200 & t3 & 0.99 & -0.4896 & 1.5882 & -0.4882 & 1.5931 \\
500 & t3 & 0.99 & -0.1830 & 1.1097 & -0.1812 & 1.1096 \\
\hline
\end{tabular}
\end{table}

\section{Numerical Illustration}
\label{sec:num}

To complement the simulation evidence, we provide a brief numerical illustration based on
real hydrological data. We use daily mean discharge observations (QD) published by the
Czech Hydrometeorological Institute (CHMI), the national authority responsible for
hydrological and meteorological monitoring in the Czech Republic.  The analyzed record comes from the CHMI gauge
\emph{Labsk\' y d\r{u}l} on the \emph{Labe} river, identified by WIGOS ID \texttt{0-203-1-000400}
(station ID 000400). The available historical daily series covers 2021--2024. To stabilize
variability and reduce the impact of occasional extreme levels, we work with the
transformed series $\log(1+QD)$.

We model the transformed series by an AR(1) nuisance autoregression and estimate its slope
using the rank-based $R$-estimator obtained via minimization of a Jaeckel-type dispersion
with the step score $J_{\lambda}(u)=\lambda-\mathbf{1}(u<\lambda)$, with $\lambda=0.5$.
Based on the effective sample size $n_{\mathrm{eff}}=1167$, the fitted persistence is
$\hat\phi_{1}=0.8712$, indicating substantial short-term dependence in the daily discharge
dynamics. Following the proposed plug-in strategy, we compute tail-risk measures directly
from the \emph{raw residuals} of this fitted autoregression.

The resulting residual tail-risk estimates equal
$\widehat{CVaR}_{0.95}=0.3352$ and $\widehat{CVaR}_{0.99}=0.5855$ on the
$\log(1+QD)$ scale, illustrating a marked increase in extreme-risk intensity when moving
from the 95\% to the 99\% level. Table~\ref{tab:chmu_num} summarizes the numerical results
and Fig.~\ref{fig:chmu_num} highlights dates corresponding to raw residuals exceeding
$\widehat{VaR}_{0.99}$.

\begin{table}[t]
\centering
\caption{Rank-based AR(1) slope estimate and raw-residual tail-risk estimates for the CHMI
daily discharge series at \emph{Labsk\' y d\r{u}l} (Labe), 2021--2024, on the $\log(1+QD)$ scale.}
\label{tab:chmu_num}
\begin{tabular}{llrrrrr}
\hline
Gauge (river) & Period & $p$ & $n_{\mathrm{eff}}$ & $\hat\phi_{1}$ &
$\widehat{CVaR}_{0.95}$ & $\widehat{CVaR}_{0.99}$ \\
\hline
Labsk\' y d\r{u}l (Labe) & 2021--2024 & 1 & 1167 & 0.8712 & 0.3352 & 0.5855 \\
\hline
\end{tabular}
\end{table}

\begin{figure}[t]
\centering
\includegraphics[width=\textwidth]{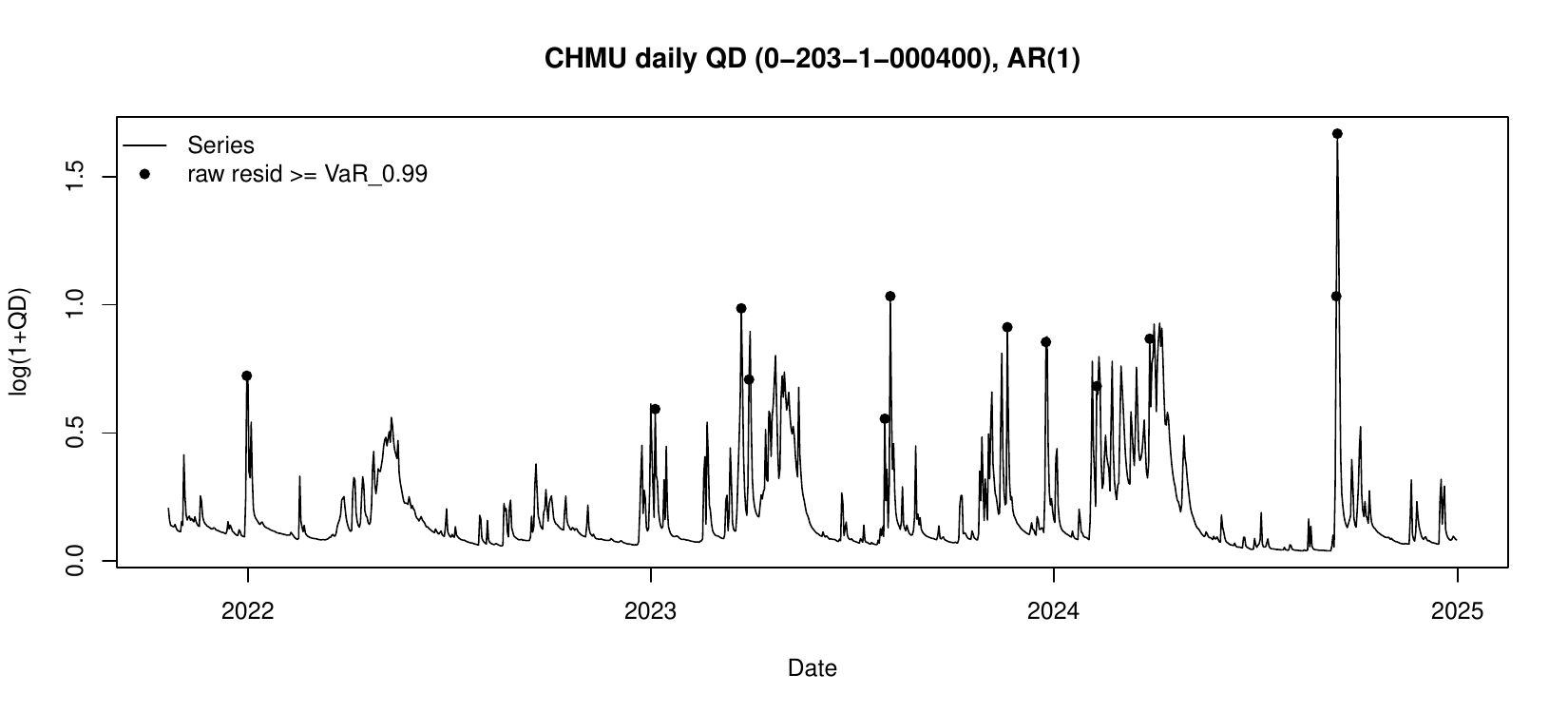}
\caption{CHMI daily mean discharge series (QD) at \emph{Labsk\' y d\r{u}l} (Labe), 2021--2024,
shown on the transformed scale $\log(1+QD)$. Filled points mark dates for which the
\emph{raw} AR(1) residuals obtained from the rank-based $R$-estimated slope exceed the
empirical $\widehat{VaR}_{0.99}$.}
\label{fig:chmu_num}
\end{figure}

\section*{Conclusion.}
We considered estimation of linear quantile functionals of the innovation distribution in
autoregressive models with unknown slope parameters, where the innovations themselves are
unobservable. The proposed construction combines autoregression quantiles with rank-based
$R$-estimation of the nuisance autoregression and yields a simple plug-in procedure for
risk measures such as $\mathrm{CVaR}_{\alpha}$. The simulation study indicates that the
finite-sample cost of estimating the nuisance AR component is small: the feasible
$R$-based estimator closely tracks the oracle benchmark even under heavy tails and
contamination. A brief numerical illustration with CHMI daily discharge data from the
Labsk\' y d\r{u}l gauge on the Labe river further supports the practical relevance of the method.
The approach is broadly applicable to risk and other tail-sensitive indicators in settings
where serial dependence must be accounted for and innovation-level inference is of primary
interest.

\section*{Acknowledgements}
This work was supported by the Czech Science Foundation under Grant 22-03636S.

\end{document}